\documentclass[sglanonrev]{neutralarticle}
\RequirePackage{tgtermes}
\RequirePackage{newtxtext}
\RequirePackage{newtxmath}
\RequirePackage{bm}
\RequirePackage{endnotes}

\OneAndAHalfSpacedXII 

\usepackage{braket}
\usepackage{booktabs}

\usepackage{algorithm}
\usepackage{algpseudocode}
\usepackage{tikz}

\usepackage{natbib}
 \bibpunct[, ]{(}{)}{,}{a}{}{,}%
\EquationsNumberedThrough    

\TheoremsNumberedThrough     
\ECRepeatTheorems  %

\MANUSCRIPTNO{}

\definecolor{myblack}{RGB}{0, 20, 114}
\usepackage[hidelinks,colorlinks=true,urlcolor=myblack,linkcolor=myblack,citecolor=myblack]{hyperref}
\usepackage{amsfonts}
\usepackage{amsmath} 
\allowdisplaybreaks[4]
\usepackage{graphicx} 
\usepackage{natbib}
\usepackage{algorithm}
\usepackage{booktabs}
\usepackage{algpseudocode}
\usepackage{amssymb}
\usepackage{array}

\def\EMAIL#1{\href{mailto:#1}{#1}}

\def\theARTICLETOP{}

\makeatletter
\def\theARTICLEABSTRACT{%
  \HOOKb
  \vspace*{18pt}
  \hspace*{2pt}\begin{minipage}[t]{\textwidth}\parindent1em
    \ABSfont
    \noindent\theABSTRACT\endgraf
    \vskip5pt
    \theFUNDING
    \theKEYWORDS
    \theSUBJECTCLASS
    \theAREAOFREVIEW
    \theMSCCLASS
    \theORMSCLASS
    \if@BLINDREV\else\theHISTORY\fi
    \noindent\hrulefill
  \end{minipage}
  \vspace*{0pt}
}
\makeatother
\begin{document}

\RUNAUTHOR{Hu and Hu}

\RUNTITLE{Rate-Optimal Quantum Discrete Simulation Optimization}

\TITLE{Rate-Optimal Quantum Discrete Simulation Optimization}

\ARTICLEAUTHORS{%
\AUTHOR{Mingjie Hu}
\AFF{School of Management, Fudan University, \EMAIL{23110690009@m.fudan.edu.cn}
}

\AUTHOR{Jian-Qiang Hu}
\AFF{School of Management, Fudan University, \EMAIL{hujq@fudan.edu.cn}}

}

\ABSTRACT{%
Quantum computing has developed rapidly in recent years and has shown potential advantages over classical computing in various domains. In this paper, we study whether quantum computing can accelerate the convergence rate of discrete simulation optimization in the fixed-budget setting. We first formulate the quantum discrete simulation optimization problem by defining the quantum simulation oracle, the algorithmic performance measure, and the corresponding optimality guarantee. Based on this formulation, we establish a minimax lower bound on the expected optimality gap for any quantum algorithm. The lower bound analysis reduces the discrete simulation optimization problem to a sequential quantum phase-testing problem, and the resulting lower-bound argument may be of independent interest. By comparing the quantum lower bound with its classical counterpart, we show that the optimal convergence rate in the quantum setting exhibits a quadratic improvement in its dependence on the simulation budget, at the cost of a worse dependence on the number of decisions. We then propose the first fixed-budget quantum algorithm, called \textsc{QUEST}, for discrete simulation optimization. The algorithm combines a quantum Monte Carlo subroutine with a sequential elimination framework and uses a slow-measurement mechanism to help preserve quantum superposition during the optimization process. We prove that \textsc{QUEST} achieves the optimal convergence rate up to logarithmic factors. Numerical experiments further show that \textsc{QUEST} outperforms its classical counterpart, which provides empirical evidence for the quantum advantage in improving the convergence rate of discrete simulation optimization algorithms.
}%


\maketitle


\section{Introduction}\label{sec:Intro}
Quantum computing has attracted rapidly growing attention in recent years, driven by significant advances in quantum hardware, algorithms, and theoretical foundations. In particular, quantum advantage has been theoretically demonstrated in several application domains, including simulation \citep{blanchet2025quantum,zabinsky2025comprehensive}, reinforcement learning \citep{meyer2022survey}, and combinatorial optimization \citep{abbas2024challenges}. 

Previous studies have shown that quantum computing can provide a quadratic speedup for Monte Carlo estimation in various simulation settings \citep{montanaro2015quantum,blanchet2025quantum}. This raises a natural question of whether similar improvement can be achieved in simulation optimization. However, the role of quantum  computing in simulation optimization remains much less understood, with limited results on both its fundamental performance limits and algorithm design. This motivates the central questions of this paper: \begin{center}
\textit{How does quantum computing affect the convergence rate of discrete simulation optimization, and can we design a quantum algorithm that is provably rate optimal?}
\end{center}

To answer these questions, we focus on discrete simulation optimization, a widely studied problem class also known as ranking and selection (R\&S) in the simulation literature. In this setting, a finite set of candidate decisions is given, and the performance of each decision is represented by a random simulation output. The goal is to recommend a decision with the largest mean performance using a limited simulation budget. Under the fixed-budget setting, we measure the quality of a recommendation by its expected optimality gap, defined as the expected difference between the mean performance of the optimal decision and that of the recommended decision. Thus, the central objective is to minimize this expected optimality gap under a given number of simulation budget.

However, addressing these theoretical and algorithmic questions is challenging for several reasons:
\begin{itemize}
    \item \textbf{Quantum Information Extraction.} The advantage of quantum computing stems from its ability to process information in superposition. However, extracting classical information requires measuring the quantum state, which collapses the superposition and prevents its further use. This restriction on information extraction makes classical sequential simulation optimization algorithms unsuitable in the quantum setting and calls for a different algorithmic structure. 

    \item \textbf{Adaptive Algorithm Design.} Quantum simulation optimization algorithms rely on repeated queries to simulation oracles to estimate the performance of candidate decisions. To use these queries efficiently, the algorithm should adaptively eliminate clearly suboptimal decisions and concentrate the remaining budget on more competitive ones. The main challenge is to design an adaptive procedure that balances sequential elimination with the limited opportunities for information extraction imposed by quantum measurement.

    \item \textbf{Performance Limit Analysis.} Performance limits for classical algorithms are often established through change-of-measure arguments from sequential hypothesis testing \citep{lattimore2020bandit}. These arguments do not directly extend to the quantum setting. In particular, because the number of oracle queries made by an adaptive quantum algorithm is random, new tools from quantum hypothesis testing are needed both to characterize the fundamental performance limit and to establish performance guarantees for the proposed algorithm.

\end{itemize}

To aid readers, we provide an intuitive explanation of the key quantum computing concepts used in this paper and clarify their connection to classical discrete simulation optimization. In the quantum setting, a simulation system is represented by a quantum simulation oracle, which serves as the analogue of a classical simulation model. Given an input decision, a classical simulator produces a random observation, while a quantum simulation oracle prepares a quantum sample: a superposition state that coherently encodes all possible simulation outcomes. Measuring this state yields an observable sample, but also collapses the superposition and destroys the quantum parallelism that enables acceleration. Under this analogy, the number of classical simulation budget corresponds to the number of queries to the quantum simulation oracle. Therefore, a key principle in quantum algorithm design is to preserve superposition as much as possible and measure only when necessary, which allows the algorithm to exploit quantum coherence and achieve faster convergence rate.

\subsection{Main Contributions}
The main contributions of this paper are summarized as follows:
\begin{itemize}
    \item We first formulate quantum discrete simulation optimization in the fixed-budget setting and introduce the quantum simulation oracle as the basic building block for quantum algorithm design. We then define the minimax expected optimality gap and the corresponding notion of minimax rate optimality for evaluating the performance of quantum simulation optimization algorithms.
    \item We reduce the discrete simulation optimization to a sequential quantum phase-testing problem and derive a lower bound on the testing error. This result provides the key step in characterizing the optimal convergence rate in the quantum setting. The sequential phase-testing argument may also be of independent interest and could be useful for studying other performance measures in quantum simulation optimization and related problems.
    \item We establish minimax lower bounds for the expected optimality gap in both quantum and classical settings. Specifically, we show that the quantum minimax lower bound is \(\Omega\left({|\mathcal X|}/{T}\right)\), while the corresponding classical lower bound is \(\Omega(\sqrt{{|\mathcal X|}/{T}})\), where $|\mathcal X|$ denotes the number of decisions and $T$ the simulation budget. The comparison reveals a quadratic improvement in the dependence of the fundamental performance limit on the simulation budget, from \(T^{-1/2}\) in the classical setting to \(T^{-1}\) in the quantum setting. At the same time, the dependence on the number of decisions becomes worse, changing from \(\sqrt{|\mathcal X|}\) to \(|\mathcal X|\).
    We provide an intuitive explanation for this trade-off.
    \item We propose a quantum discrete simulation optimization algorithm, called \textsc{QUEST}. The algorithm follows a sequential elimination framework, using quantum Monte Carlo estimation to evaluate candidate decisions and progressively eliminates decisions with inferior estimated performance. To adapt sequential elimination to the quantum setting, we introduce a slow-measurement mechanism that helps preserve quantum coherence and retain the potential quantum speedup. We prove that \textsc{QUEST} achieves an expected optimality gap of order \(O((|\mathcal X|\log T)/T)\). By comparing this upper bound with the quantum minimax lower bound \(\Omega(|\mathcal X|/T)\), we show that \textsc{QUEST} is minimax rate-optimal up to logarithmic factors.
    \item Finally, we conduct numerical experiments to illustrate the quantum advantage. We implement \textsc{QUEST} using Qiskit-based quantum simulation and compare its performance with that of a classical counterpart. The results show that, under the same simulation budget, \textsc{QUEST} achieves a smaller expected optimality gap. This empirical observation is consistent with our theoretical convergence-rate analysis.
\end{itemize}

These results show that quantum computing can improve the convergence rate of simulation optimization and that this improvement is attainable by a provably near-optimal algorithm. More broadly, our framework opens new directions for studying the role of quantum computing in simulation-based decision-making.

\subsection{Related Literature}
Discrete simulation optimization, also known as R\&S, is a widely studied class of simulation optimization problems. Two standard formulations are the fixed-budget and fixed-confidence settings. In the fixed-budget setting, \citet{chen2000simulation} proposed the optimal computing budget allocation (OCBA) procedure, a foundational approach for allocating a limited simulation budget to identify the decision with the best mean performance. Subsequent work has extended this framework to problems with spatial structure \citep{hu2025optimal}, stochastic constraints \citep{lee2012approximate,hu2024multi}, and input uncertainty \citep{wang2025optimal}.
In the fixed-confidence setting, \citet{kim2001fully} proposed the KN procedure, a fully sequential method that guarantees the probability of correct selection to exceed a prescribed confidence level. Related fixed-confidence formulations have since been studied in settings with contextual information \citep{shen2021ranking} and data-driven problem formulations \citep{wu2024data}.

Recent work has begun to examine how quantum computing can be used to improve simulation optimization \citep{zabinsky2025comprehensive}. The most closely related studies are \citet{hu2025quantum} and \citet{hu2026quantum}, which studied discrete simulation optimization in the fixed-confidence setting and focused on reducing query complexity. We instead study the fixed-budget setting and characterize performance through the minimax expected optimality gap. The two formulations lead to different objectives, performance measures, and algorithmic challenges, so the theoretical analysis and algorithm design required here are substantially different from those developed for fixed-confidence problems.

Our study is also connected to the quantum search literature. Grover's algorithm achieves a quadratic speedup in query complexity for finding a marked item \citep{grover1996fast}, and the D\"urr-H{\o}yer minimum-finding algorithm extends this idea to finding the minimum of an unsorted list \citep{durr1996quantum}. Grover adaptive search and related methods further combined threshold updates with amplitude amplification to search for decisions that improve on the current objective value \citep{bulger2003implementing,baritompa2005grover,gilliam2021grover}. These methods primarily address deterministic objective functions, while our setting explicitly accounts for simulation noise. The oracle models are also different. Existing quantum search methods typically assume coherent access to an oracle acting on a superposition of candidate decisions, while our framework associates a quantum simulation oracle with each decision. Moreover, existing search-based work is usually evaluated through the query complexity required to identify an optimum. We instead study the convergence rate of the expected optimality gap under a fixed simulation budget. This performance measure is standard in simulation optimization \citep{gao2017new}, but its fundamental limit in the quantum setting remains largely unexplored.

The algorithmic development builds on quantum mean estimation, which improves the query complexity of mean estimation quadratically over classical Monte Carlo methods \citep{montanaro2015quantum,kothari2023mean}. Applying mean estimation independently to every decision, however, can spend many queries on decisions that are already clearly suboptimal. The main algorithmic question is how to allocate queries adaptively, eliminating inferior decisions while concentrating effort on the remaining competitive ones. This type of adaptive allocation is central to classical simulation optimization \citep{wu2024data}. \citet{gacon2020quantum} developed a quantum-enhanced framework for simulation-based optimization by combining quantum amplitude estimation with optimization procedures. Our focus is different since we establish a minimax lower bound on the convergence rate and design an adaptive algorithm that achieves this rate up to logarithmic factors.

Quantum acceleration has also been studied in multi-armed bandits \citep{casale2020quantum}, a closely related problem in machine learning. Most studies focus on cumulative regret \citep{wan2023quantum, dai2023quantum},  which accounts for losses incurred throughout the learning process. The expected optimality gap considered here depends only on the final recommendation. Other studies on quantum best-arm identification have adopted fixed-confidence objectives by controlling the probability of incorrect selection and reducing query complexity \citep{wang2021quantum,wang2025quantum}. In contrast, we study the convergence rate of quantum discrete simulation optimization under a fixed budget. Our lower-bound analysis also develops a sequential phase-testing tool that accommodates adaptive quantum algorithms with random query counts.

The remainder of this paper is organized as follows. Section~\ref{sec: problem} formally formulates the discrete simulation optimization problem and introduces the quantum computing preliminaries and the definition of the quantum simulation oracle. Section~\ref{sec: rate bound} defines minimax rate optimality and establishes minimax lower bounds on the expected optimality gap for both quantum and classical algorithms. Section~\ref{sec: algorithm} proposes a quantum algorithm for discrete simulation optimization and proves its rate optimality up to logarithmic factors. Section~\ref{sec: numerical} reports the numerical results. Section~\ref{sec: conclusion} concludes the paper and discusses future research directions. All technical proofs are provided in the appendix.

\section{Quantum Discrete Simulation Optimization}
\label{sec: problem}
In this section, we first introduce the mathematical formulation of discrete simulation optimization. We then present the quantum computing preliminaries used in this paper; for a comprehensive treatment, see \citet{nielsen2010quantum}. Finally, we define the quantum simulation oracle, which serves as a fundamental building block for designing quantum algorithms.

\subsection{Discrete Simulation Optimization}
We consider a complex stochastic system whose performance depends on a discrete decision variable $x$ selected from a feasible set $\mathcal X$. The system is characterized by three components: a decision variable $x\in \mathcal X$; a random object $\xi_x$, defined on a support set $\Xi$ with distribution $P(\xi_x)$, which summarizes all random quantities and processes associated with the system under decision $x$; and a deterministic function $Y:\mathcal{X}\times \Xi\rightarrow\mathbb{R}$, which maps the decision and a realization of the randomness to the corresponding system performance. In particular, $Y$ captures the operational logic of the stochastic system and does not admit a closed-form expression. We define the objective value of decision $x$ as its mean performance:
\begin{equation*}
    y(x) := \mathbb{E}[Y(x,\xi_x)].
\end{equation*}
The objective value $y(x)$ generally has no closed-form representation and can only be estimated from simulation outputs. Let \(x^\star\) denote an optimal decision, i.e., any decision satisfying
\begin{equation*}
    y(x^\star) = \max_{x\in\mathcal X} y(x).
\end{equation*}

Let \(\mathcal A\) denote a discrete simulation optimization algorithm, and let
\(\nu=\{P_x:x\in\mathcal X\}\) denote a given problem instance, where \(P_x\) is the distribution of the simulation output \(Y(x,\xi_x)\) under decision \(x\). Definition~\ref{def: expected optgap} introduces the expected optimality gap, which we use to measure the performance of algorithm \(\mathcal A\) under instance \(\nu\).

\begin{definition}[Expected Optimality Gap]
\label{def: expected optgap}
    For a given simulation budget \(T\), the expected optimality gap of an algorithm \(\mathcal A\) under problem instance \(\nu\) is defined as
    \begin{equation}
\label{eq: opt gap}
    \mathcal{E}_T(\mathcal A,\nu)
    :=
    \mathbb{E}_{\mathcal A,\nu}
    \left[
        y(x^\star)-y(\hat{x}_T)
    \right],
\end{equation}
where \(\hat{x}_T\) is the decision recommended by algorithm \(\mathcal A\) under simulation budget \(T\).
\end{definition}
\begin{remark}
The expectation in \eqref{eq: opt gap} is taken over the simulation randomness and any randomization used by the algorithm. The expected optimality gap is a standard performance measure in the literature. It is also referred to as the expected opportunity cost in the R\&S literature \citep{gao2017new} and as simple regret in the multi-armed bandit literature \citep{lattimore2020bandit}.
\end{remark}

To facilitate the exposition and theoretical analysis, we impose the following assumption on the simulation outputs.

\begin{assumption}
\label{ass:performance}
For every \(x\in\mathcal X\), the centered simulation output
\(Y(x,\xi_x)-y(x)\) is \(\sigma\)-sub-Gaussian; that is, for all
\(\lambda\in\mathbb R\),
\begin{equation*}
    \mathbb E\left[
        \exp\left(\lambda\big(Y(x,\xi_x)-y(x)\big)\right)
    \right]
    \le
    \exp\left(\frac{\lambda^2\sigma^2}{2}\right),
\end{equation*}
and $\sigma\ge1/2$ is a fixed constant independent of $T$ and $|\mathcal X|$.
Moreover, the mean performance satisfies \(y(x)\in[0,1]\) for all
\(x\in\mathcal X\).
\end{assumption}

\begin{remark}
The sub-Gaussian condition is standard and generalizes the Gaussian noise
assumption commonly adopted in the classical simulation optimization
literature \citep{chen2000simulation}. The assumption \(y(x)\in[0,1]\) is a normalization made without loss of generality for analysis. Whenever the mean performances are uniformly bounded, one can apply an affine rescaling to the simulation outputs so that the normalized mean performances lie in \([0,1]\).
\end{remark}

\subsection{Preliminaries of Quantum Computing}
Quantum computers exploit quantum phenomena such as superposition, entanglement, and interference to process information. These principles enable quantum parallelism, in which unitary operations act coherently on an exponential superposition of classical basis states. Intuitively, quantum parallelism can be viewed as evaluating a function at multiple input values simultaneously. For certain computational tasks, quantum computers can achieve exponential speedups over classical computers. In this subsection, we introduce the key concepts and notation used throughout the paper. For a comprehensive treatment, see \citep{nielsen2010quantum}.

In quantum computing, Dirac's bra-ket notation is used to represent vectors in complex Hilbert spaces. A ket
$
|\psi\rangle = (\psi_1,\dots,\psi_n)^\top \in \mathbb{C}^n
$
denotes a column vector, and its corresponding bra
$
\langle \psi| = (\psi_1^*,\dots,\psi_n^*)
$
is the conjugate transpose of the ket. For any two states
\(|\psi\rangle,|\phi\rangle \in \mathbb{C}^n\), their inner product is defined as
\[
\langle \psi|\phi\rangle
:=
\sum_{i=1}^n \psi_i^* \phi_i
\in \mathbb{C},
\]
which induces the \(\ell_2\)-norm
\[
\big\| |\psi\rangle \big\|_2
=
\sqrt{\langle \psi|\psi\rangle}.
\]
The tensor product of two quantum states
\(|\psi\rangle \in \mathbb{C}^{n_1}\) and
\(|\phi\rangle \in \mathbb{C}^{n_2}\) is defined as
\[
|\psi\rangle|\phi\rangle
:=
|\psi\rangle \otimes |\phi\rangle
=
(\psi_1\phi_1,\psi_1\phi_2,\dots,\psi_1\phi_{n_2},
\psi_2\phi_1,\dots,\psi_{n_1}\phi_{n_2})^\top
\in
\mathbb{C}^{n_1}\otimes\mathbb{C}^{n_2},
\]
and represents the joint state of the composite system.

The basic unit of information in quantum computing is a qubit. Unlike a classical bit, which can take one of two values, \(0\) or \(1\), a qubit can exist in a linear combination of the two basis states, known as a quantum superposition. Mathematically, a single-qubit state can be written as
$
    \ket{\psi} = \alpha \ket{0} + \beta \ket{1},
$
where \(\alpha,\beta \in \mathbb{C}\) are complex amplitudes satisfying the normalization condition
$
    |\alpha|^2 + |\beta|^2 = 1.
$
Upon measurement in the computational basis, the qubit collapses to \(\ket{0}\) with probability \(|\alpha|^2\) or to \(\ket{1}\) with probability \(|\beta|^2\). 

\subsection{Quantum Simulation Oracle}
To leverage quantum computing for accelerating simulation optimization, we introduce a quantum simulation oracle for each decision as the basic building block for quantum algorithm design. To make the construction clear, we first formulate the classical simulation oracle and then extend this notion to the quantum setting.

In classical simulation optimization, a simulation model is used to estimate the expected system performance of each decision. Given an input \(x\in\mathcal X\), one simulation replication generates a random observation \(Y(x,\xi_x)\). Thus, the simulation model can be viewed as a stochastic oracle that maps a decision \(x\) to a realized performance outcome \(Y(x,\xi_x)\). We refer to this oracle as the classical simulation oracle.

We define the quantum simulation oracle analogously. To estimate system performance, we construct a quantum simulation model. Intuitively, querying this model with an input \(x\in\mathcal X\) produces a quantum sample of \(Y(x,\xi_x)\). Unlike a classical simulation output, however, this sample is encoded in a superposition state and cannot be directly observed before measurement.

Definition~\ref{def: quantum sim oracle} gives the formal definition of the quantum simulation oracle for each decision $x\in \mathcal X$. We use a simulation noise register \(\xi\) and an output register \(Y\), both initialized in the zero state.

\begin{definition}[Quantum Simulation Oracle]
\label{def: quantum sim oracle}
For each decision $x\in\mathcal X$, a quantum simulation oracle is a unitary operator \(\mathcal Q_x\) such that
\begin{equation}
\label{eq: quantum oracle}
\mathcal Q_x
\ket{0}_{\xi}\ket{0}_{Y}
=
    \sum_{\xi_x\in \Xi}
    \sqrt{p_x(\xi_x)}
    \ket{\xi_x}_{\xi}
    \ket{Y(x,\xi_x)}_{Y},
\end{equation}
where \(\Xi\) is the support of the simulation noise \(\xi_x\), and \(p_x(\xi_x)\) is its probability mass function. The induced distribution of $Y(x,\xi_x)$ is $P_x$.
\end{definition}
\begin{remark}
The quantum simulation oracle in \eqref{eq: quantum oracle} is stated for discrete random variables. 
For continuous random variables, one first discretizes the support into finitely many grid points, so that the distribution can be encoded on a finite number of qubits. 
Thus, the prepared quantum state represents a discretized approximation of the original continuous simulation model \citep{carrera2021efficient}.
\end{remark}

For the quantum simulation oracle, each query prepares one quantum sample encoded in superposition. Thus, the number of classical simulation replications is naturally compared with the number of queries to the quantum simulation oracle. From this perspective, the quantum simulation process and the corresponding definition of the optimality gap are aligned with their classical counterparts.

The quantum simulation oracle is also implementable in principle. If a classical simulator, such as a queueing model, can be implemented as a computer program, then it can be converted into a quantum circuit with at most polynomial overhead \citep{wang2021quantum2,nielsen2010quantum}. Existing studies have constructed quantum simulation models for specific applications, including queueing systems and reinforcement learning; see, for example, \citet{esa2024single} and \citet{dunjko2016quantum}. Since this paper focuses on designing rate-optimal simulation optimization algorithms, we assume that such a quantum simulation model is available.

\section{Optimal Rate of Expected Optimality Gap}
\label{sec: rate bound}
In this section, we first formally introduce the minimax expected optimality gap, along with the notions of minimax optimality and minimax rate optimality. We then establish the optimal rates of the expected optimality gap for quantum and classical algorithms, respectively. These results characterize the potential quantum advantage in terms of convergence rates.

We introduce the standard Bachmann-Landau notation in Definition~\ref{def:landau} \citep{lattimore2020bandit}.

\begin{definition}[Landau Notation]
\label{def:landau}
Let \(f,g:\mathbb N\to[0,\infty)\) be functions with \(g(n)>0\) for all sufficiently large \(n\). We write
\begin{equation*}
    f(n)=O(g(n))
    \Longleftrightarrow
    \limsup_{n\to\infty}\frac{f(n)}{g(n)}<\infty,
\end{equation*}
and
\begin{equation*}
    f(n)=\Omega(g(n))
    \Longleftrightarrow
    \liminf_{n\to\infty}\frac{f(n)}{g(n)}>0.
\end{equation*}
\end{definition}

\begin{remark}
In this paper, \(f(n)=O(g(n))\) means that \(f(n)\) grows no faster than \(g(n)\) up to constant factors, while \(f(n)=\Omega(g(n))\) means that \(f(n)\) grows at least as fast as \(g(n)\) up to constant factors. When the upper and lower bounds match asymptotically, we say that the corresponding rate is asymptotically optimal.
\end{remark}

To evaluate the performance of a simulation optimization algorithm, it is not sufficient to consider a single problem instance. An algorithm may perform well on some instances but poorly on others. Therefore, we evaluate algorithms over a class of admissible problem instances. Let \(\mathcal M\) denote the class of admissible discrete simulation optimization problem instances satisfying Assumption~\ref{ass:performance}, i.e.,
\begin{equation*}
    \mathcal M
    =
    \left\{
        \nu=(P_x:x\in\mathcal X):
        \begin{array}{l}
        \mathbb E_{Z\sim P_x}
        \left[
            \exp\left(\lambda\big(Z-y(x)\big)\right)
        \right]
        \le
        \exp\left(\lambda^2\sigma^2/2\right),
        \quad \forall \lambda\in\mathbb R,\ \forall x\in\mathcal X, \\[0.5em]
        y(x):=\mathbb E_{Z\sim P_x}[Z]\in[0,1],
        \quad \forall x\in\mathcal X
        \end{array}
    \right\}.
\end{equation*}

To characterize the performance of an algorithm over a class of problem instances, we adopt a worst-case perspective. Specifically, for a fixed algorithm \(\mathcal A\), we evaluate its largest expected optimality gap over all instances \(\nu\in\mathcal M\). We then compare this worst-case performance with the best possible worst-case performance achievable by any admissible quantum algorithm. This leads to the following definition of the minimax expected optimality gap.

\begin{definition}[Minimax Expected Optimality Gap]
\label{def:minimax-expected-optgap}
For a given simulation budget \(T\), the minimax expected optimality gap over the instance class \(\mathcal M\) is defined as
\begin{equation}
\label{eq:minimax-gap}
\mathcal{E}^\star_T(\mathcal M)
:=
\inf_{\mathcal A}
\sup_{\nu\in \mathcal M}
\mathcal{E}_T(\mathcal A,\nu),
\end{equation}
where the infimum is taken over all admissible quantum algorithms \(\mathcal A\).
\end{definition}
\begin{remark}
The quantity \(\mathcal{E}^\star_T(\mathcal M)\) provides a theoretical benchmark for evaluating algorithms over the instance class \(\mathcal M\). It is also of independent interest, as it characterizes the best achievable worst-case performance among all admissible quantum simulation optimization algorithms over \(\mathcal M\).
\end{remark}

Based on the minimax expected optimality gap, we define minimax optimality for a quantum simulation optimization algorithm. This notion requires that the worst-case expected optimality gap of an algorithm over the instance class \(\mathcal M\) attain the minimax benchmark \(\mathcal E^\star_T(\mathcal M)\). In other words, a minimax optimal algorithm achieves the best possible worst-case performance among all admissible quantum algorithms under the same simulation budget \(T\).

\begin{definition}[Minimax Optimality]
A quantum discrete simulation optimization algorithm \(\mathcal A\) is minimax optimal over \(\mathcal M\) under simulation budget \(T\) if
\begin{equation*}
    \sup_{\nu\in\mathcal M}
    \mathcal E_T(\mathcal A,\nu)
    =
    \mathcal E^\star_T(\mathcal M).
\end{equation*}
\end{definition}

Although \(\mathcal E_T^\star(\mathcal M)\) characterizes the fundamental worst-case performance limit over the instance class \(\mathcal M\), requiring an algorithm to attain this value exactly is often too stringent. Therefore, we consider the weaker and more practical notion of minimax rate optimality in Definition \ref{def: rate opt}. Under this notion, an algorithm is considered optimal if the minimax expected optimality gap admits a lower bound of a certain rate, and the worst-case expected optimality gap of the algorithm achieves a matching upper bound up to constant factors.

\begin{definition}[Minimax Rate Optimality]
\label{def: rate opt}
A quantum discrete simulation optimization algorithm \(\mathcal A\) is minimax rate optimal over \(\mathcal M\) if there exists a rate \(r_T\) such that the minimax expected optimality gap is lower bounded by \(r_T\) up to a constant factor, and the worst-case expected optimality gap of \(\mathcal A\) is upper bounded by \(r_T\) up to a constant factor. Formally,
\begin{equation*}
\mathcal E_T^\star(\mathcal M)
=
\Omega(r_T),
\qquad
\sup_{\nu\in\mathcal M}
\mathcal E_T(\mathcal A,\nu)
=
O(r_T).
\end{equation*}
\end{definition}

\begin{remark}
In some cases, the upper bound achieved by an algorithm matches the minimax lower bound up to logarithmic factors. Since logarithmic factors are lower order relative to the polynomial dependence on the simulation budget \(T\), we refer to such an algorithm as minimax rate optimal up to logarithmic factors. We use the notation \(\widetilde O(\cdot)\) to hide logarithmic factors. Thus, if \(\mathcal E_T^\star(\mathcal M)=\Omega(r_T)\) and
\[
    \sup_{\nu\in\mathcal M}\mathcal E_T(\mathcal A,\nu)
    =
    \widetilde O(r_T),
\]
then algorithm \(\mathcal A\) is minimax rate optimal up to logarithmic factors.
\end{remark}

Computing the exact value of \(\mathcal E_T^\star(\mathcal M)\) is generally intractable, because it requires optimizing over all admissible quantum algorithms and taking the worst case over the entire instance class \(\mathcal M\). Therefore, we focus on the notion of minimax rate optimality introduced in Definition~\ref{def: rate opt}. To derive the lower-bound part of the minimax rate analysis, we restrict attention to a Bernoulli subclass whose instances satisfy the assumptions defining \(\mathcal M\). Since this subclass is contained in \(\mathcal M\), any minimax lower bound for the subclass is also a valid minimax lower bound for the full class \(\mathcal M\).

The minimax lower-bound analysis relies on a reduction to hypothesis testing: if two problem instances are difficult to distinguish, then no algorithm can reliably identify the better decision under a limited simulation budget. In the classical setting, such arguments are typically based on change-of-measure or information-theoretic bounds for probability distributions \citep{lattimore2020bandit}. In the quantum setting, however, these classical hypothesis testing tools do not directly apply, because the information available to an algorithm is encoded in quantum states rather than classical samples. 

To establish a minimax lower bound on the expected optimality gap, a standard approach is to construct a suitable family of problem instances within $\mathcal M$ on which any algorithm must incur a nontrivial optimality gap \citep{lattimore2020bandit}. Following this approach,
 we construct Bernoulli instances in which the simulation output under decision $x$ satisfies
\[
    \mathbb P(Y(x)=1)=y(x),
    \qquad
    \mathbb P(Y(x)=0)=1-y(x).
\]
This is a special case of Definition~\ref{def: quantum sim oracle} with
$\Xi=\{0,1\}$, $Y(x,\xi)=\xi$, $p_x(1)=y(x)$, and $p_x(0)=1-y(x)$. One quantum simulation oracle
query outputs
\[
    \mathcal Q_x\ket{0}_{\xi}\ket{0}_{Y}
    =
        \sqrt{1-y(x)}\ket{0}_{\xi}\ket{0}_{Y}
        +
        \sqrt{y(x)}\ket{1}_{\xi}\ket{1}_{Y}.
\]
Since the decision register is fixed and the output map $Y$ is deterministic,
we identify the induced Bernoulli amplitude state as
\[
    \ket{\psi(x)}
    =
    \sqrt{1-y(x)}\ket{0}
    +
    \sqrt{y(x)}\ket{1}.
\]
Measuring $\ket{\psi(x)}$ in the computational basis yields outcome $1$ with
probability $y(x)$ and outcome $0$ with probability $1-y(x)$, which matches the
Bernoulli simulation model.

We represent the set of decisions as $\mathcal{X} = \{x_1,\ldots,x_{|\mathcal X|}\}$. 
For Bernoulli instances, the distribution $P_x$ is determined by the mean parameter. Therefore, we describe their output distributions by the
mean vector $\nu = (y(x_1),\ldots,y({x_{|\mathcal X|})})$. Set $\Delta_0 := 1/64$. For $\Delta\in(0,\Delta_0]$, define the original problem instance $\nu_1:= (1/4+\Delta,1/4,\ldots,1/4)$. Therefore, $x_1$ is the unique optimal decision. For each $i\in\{2,\ldots,|\mathcal X|\}$, let
\(\mathcal V_i(\Delta)\) denote the set of Bernoulli problem
instances \(\nu\) for which
\[
\nu=\left(
\frac{1}{4}+\Delta,\,
\frac{1}{4},\,
\ldots,\,
\underset{\text{$i$th coordinate}}{y(x_i)},\,
\ldots,\,
\frac{1}{4}
\right),
\qquad
\frac{1}{4}+2\Delta
\le y(x_i) \le
\frac{1}{4}+\frac{\sqrt{3}}{8}.
\]
Therefore, under $\nu\in \mathcal V_i(\Delta)$, the optimal decision is $x_i$ with an optimality gap of at least $\Delta$. 

Lemma~\ref{lemma: exp error} provides an explicit analytical lower bound on the quantum hypothesis-testing error probability under the Bernoulli instances. The bound shows that the error probability decays at most exponentially with the product of the number of queries and the mean gap \(\Delta\). This lower bound provides the key ingredient for the subsequent lower-bound analysis of the minimax expected optimality gap.

\begin{lemma}
\label{lemma: exp error}
There exists a universal constant $c>0$ such that, for every $0<\Delta\le \Delta_0$, every $i\in\{2,\ldots,|\mathcal X|\}$, and every quantum algorithm with recommended decision $\hat{x}$,
\[
\max\left\{
\mathbb{P}_{\nu_1}\!\left(\hat{x}\neq x_1\right),
\sup_{\nu\in\mathcal{V}_i(\Delta)}
\mathbb{P}_{\nu}\!\left(\hat{x}\neq x_i\right)
\right\}
\ge
\frac{1}{4}
\exp\left[
-c\Delta
\left(
\mathbb{E}_{\nu_1}[N_i]+\frac{1}{2}
\right)
\right].
\]
where $N_i$ is the random total number of queries directed to the oracle of decision $x_i$ under $\nu_1$.
\end{lemma}

Theorem~\ref{thm: quantum lower bound} establishes a minimax lower bound on the expected optimality gap. Specifically, under our quantum simulation oracle model, no algorithm using at most \(T\) oracle queries can achieve a worst-case expected optimality gap smaller than order \(|\mathcal{X}|/T\). This result characterizes a fundamental performance limit for quantum discrete simulation optimization over the instance class \(\mathcal{M}\).

\begin{theorem}
\label{thm: quantum lower bound}
Suppose $|\mathcal X|>1$ and $T\ge |\mathcal X|-1$. Then, for any quantum
discrete simulation optimization algorithm $\mathcal A_q$ under the quantum simulation oracle in Definition \ref{def: quantum sim oracle} and using at most $T$
oracle queries, there exists an instance $\nu\in\mathcal M$ such that
\[
    \mathcal E_T(\mathcal A_q,\nu)
    =
    \Omega\left(\frac{|\mathcal X|}{T}\right).
\]
Consequently,
\[
    \mathcal{E}^\star_T(\mathcal M)
    =
    \inf_{\mathcal A_q}\sup_{\nu\in\mathcal M}\mathcal E_T(\mathcal A_q,\nu)
    =
    \Omega\left(\frac{|\mathcal X|}{T}\right).
\]
\end{theorem}

We conclude this part by highlighting the technical novelty of Lemma \ref{lemma: exp error} and Theorem \ref{thm: quantum lower bound}. Lemma \ref{lemma: exp error} reduces discrete simulation optimization to a sequential quantum phase-testing problem and applies the polynomial method \citep{mande2026tight} to derive a lower bound on the testing error. The key idea is that a quantum algorithm cannot reliably detect a small change in the mean of a decision unless it allocates sufficiently many queries to that decision. Theorem \ref{thm: quantum lower bound} then converts this testing-error bound into a minimax lower bound on the expected optimality gap. For each suboptimal decision, we construct an alternative instance in which that decision becomes optimal. Since the algorithm has only a limited query budget, it cannot accurately distinguish the original instance $\nu_1$ from all of these alternatives with high accuracy. As a result, there must exist at least one instance on which the algorithm selects a suboptimal decision with sufficiently large probability, which leads to the desired lower bound on the expected optimality gap.

We next establish the corresponding minimax lower bound for classical algorithms. Theorem~\ref{thm: classical lower bound} shows that any classical simulation optimization algorithm using at most \(T\) simulation replications must incur a worst-case expected optimality gap of at least order \(\sqrt{|\mathcal X|/T}\). This result characterizes the fundamental classical performance limit over the instance class \(\mathcal M\).

\begin{theorem}
\label{thm: classical lower bound}
Suppose $|\mathcal X|>1$ and $T\ge |\mathcal X|-1$. Then, for any classical
discrete simulation optimization algorithm $\mathcal A_c$ using at most $T$
simulation replications, there exists an
instance $\nu\in\mathcal M$ such that
\[
    \mathcal E_T(\mathcal A_c,\nu)
    =
    \Omega\left(\sqrt{\frac{|\mathcal X|}{T}}\right).
\]
Consequently,
\[
    \inf_{\mathcal A_c}
    \sup_{\nu\in\mathcal M}
    \mathcal E_T(\mathcal A_c,\nu)
    =
    \Omega\left(\sqrt{\frac{|\mathcal X|}{T}}\right).
\]
\end{theorem}

By comparing Theorem~\ref{thm: quantum lower bound} with Theorem~\ref{thm: classical lower bound}, we obtain the first indication of a quantum advantage in the dependence on the simulation budget \(T\). The quantum minimax lower bound scales as \(\Omega\left({|\mathcal X|}/{T}\right)\), while the classical minimax lower bound scales as \(\Omega(\sqrt{{|\mathcal X|}/{T}})\). 
The dependence on \(T\) improves from \(T^{-1/2}\) in the classical setting to \(T^{-1}\) in the quantum setting. This suggests a quadratic improvement in the convergence rate with respect to the simulation budget. In particular, when \(T\) is larger than \(|\mathcal X|\), the quantum lower-bound rate decreases more rapidly as \(T\) increases, which indicates the potential for quantum algorithms to achieve a smaller expected optimality gap under the same simulation budget.

An interesting, and somewhat counterintuitive, observation is that the quantum lower bound has a stronger dependence on the number of decisions \(|\mathcal X|\) than the classical lower bound. At first glance, this may seem to suggest that quantum algorithms are more sensitive to a large decision space. However, this dependence reflects how quantum algorithms accumulate statistical information.
Consider a least favorable problem instance in which the optimal decision is only slightly better than its closest competitor. Here, \(\Delta\) denotes the small mean-performance gap between the optimal decision and this nearest competing decision. Classically, distinguishing two Bernoulli means separated by a gap of this order requires on the order of \(1/\Delta^2\) samples. Equivalently, the testing error decays exponentially in \(\mathbb{E}[N]\Delta^2\), where \(N\) is the number of samples allocated to the relevant decision. To keep the error probability bounded away from zero, one must have \(\mathbb{E}[N]\Delta^2\) of constant order. Since the total budget \(T\) is spread over \(|\mathcal X|\) decisions, the least-sampled decision receives about \(T/|\mathcal X|\) samples. This leads to the gap scale \(\Delta\asymp \sqrt{|\mathcal X|/T}\).

In the quantum setting, the corresponding hypothesis-testing error decays exponentially in \(\mathbb{E}[N]\Delta\), where \(N\) is the number of quantum queries allocated to the relevant decision. Thus, to keep the error probability bounded away from zero, it is enough to have \(\mathbb E[N]\Delta\) of constant order. Equivalently, resolving a gap \(\Delta\) requires on the order of \(1/\Delta\) expected quantum queries. When the total query budget \(T\) is spread over \(|\mathcal X|\) decisions, the least-sampled decision receives about \(T/|\mathcal X|\) expected queries, which leads to the least favorable gap scale \(\Delta\asymp |\mathcal X|/T\). Quantum queries can resolve much smaller gaps under the same budget. Therefore, the worst-case problem instances for quantum algorithms place the mean performances of competing decisions closer together, with a gap of order \(|\mathcal X|/T\) rather than \(\sqrt{|\mathcal X|/T}\).

\section{Quantum Simulation Optimization Algorithm}
\label{sec: algorithm}
In this section, we propose a quantum algorithm for discrete simulation optimization. We then establish an upper bound on its minimax expected optimality gap over the instance class \(\mathcal M\). By comparing this upper bound with the minimax lower bound developed earlier, we show that the proposed algorithm attains the optimal convergence rate up to logarithmic factors. This demonstrates the minimax rate optimality of the quantum simulation optimization algorithm.

\subsection{Quantum Algorithm Design}

To design a quantum algorithm for discrete simulation optimization and achieve quantum acceleration, we build on existing quantum techniques. The following lemma provides a key theoretical guarantee for quantum Monte Carlo estimation \citep{kothari2023mean}. Since quantum Monte Carlo methods have already been implemented in software packages such as Qiskit \citep{javadi2024quantum}, we treat quantum mean estimation as a black-box subroutine. Our focus is not on its internal implementation, but on how to use its quantum speedup and theoretical guarantees to design a simulation optimization algorithm and establish its minimax rate optimality.

\begin{lemma}[Theorem 1.1 in \citealt{kothari2023mean}]
\label{lemma:qae}
Suppose Assumption~\ref{ass:performance} holds. For any decision
\(x\in\mathcal X\), assume access to a quantum simulation oracle for the
simulation output \(Y(x,\xi_x)\). Then there exists a quantum algorithm,
denoted by \(\textsc{QMC}(\tilde\epsilon,\tilde\delta)\), that uses
\[
    O\left(
        \frac{\sigma}{\tilde\epsilon}
        \log\frac{1}{\tilde\delta}
    \right)
\]
queries to the oracle and outputs an estimate \(\hat y(x)\) such that
\[
    \mathbb P\left(
        |\hat y(x)-y(x)|\le \tilde\epsilon
    \right)
    \ge
    1-\tilde\delta .
\]
\end{lemma}
\begin{remark}
Under Assumption~\ref{ass:performance}, classical Monte Carlo estimation based on the sample mean requires \(O(\sigma^2\tilde\epsilon^{-2}\log(1/\tilde\delta))\) simulation replications to achieve accuracy \(\tilde\epsilon\) with probability at least \(1-\tilde\delta\). In contrast, Lemma~\ref{lemma:qae} shows that quantum mean estimation requires only \(O(\sigma\tilde\epsilon^{-1}\log(1/\tilde\delta))\) oracle queries, which shows a quadratic improvement in the dependence on the accuracy parameter.
\end{remark}

However, designing a quantum algorithm for discrete simulation optimization and proving its optimality involve fundamental challenges. A naive approach that replaces classical estimators with quantum estimators and then analyzes the optimality gap using Lemma~\ref{lemma:qae} is not sufficient. The quantum speedup relies on preserving coherent superposition states, while intermediate measurements would collapse these states and destroy the coherence needed for acceleration. Therefore, unlike classical adaptive fixed-budget algorithms such as OCBA \citep{chen2000simulation}, a quantum algorithm cannot frequently observe simulation outputs and adjust the sampling allocation in the same way.
Instead, we need a mechanism that limits intermediate measurements and carefully controls when information is extracted from the quantum system. Moreover, to achieve minimax optimality, the algorithm must balance the estimation accuracy and failure probability of the quantum Monte Carlo subroutine, while using an appropriate elimination and stopping rule under a fixed query budget. These requirements make both the algorithm design and the theoretical analysis fundamentally different from those in classical simulation optimization.

To address these challenges, we propose \textsc{QUEST}, a \textbf{Qu}antum \textbf{E}fficient \textbf{S}imulation op\textbf{T}imization algorithm. The algorithm follows a fixed-budget sequential elimination framework. Starting from the full decision set, \textsc{QUEST} maintains an active set of candidate decisions and a decreasing threshold. In each iteration, it applies the quantum Monte Carlo subroutine to estimate the mean performance of every active decision at the same accuracy level, eliminates decisions whose estimated performance is sufficiently below the current best estimate, and then reduces the threshold by a factor of two. The procedure continues until the query budget is exhausted or only one candidate remains, after which any remaining active decision is returned as the recommended solution.

The design of \textsc{QUEST} is tailored to the quantum setting. Instead of repeatedly measuring individual simulation outputs and adaptively updating the query allocation after each observation, the algorithm queries the quantum simulation oracle in batches through the quantum Monte Carlo subroutine and extracts estimates only at the end of each iteration. This slow-measurement structure is compatible with the need to preserve quantum coherence during the estimation procedure. Moreover, the estimation accuracy \(\tilde\epsilon_r=\varepsilon_r/2\), the per-estimate failure probability \(\delta=1/T\), and the query allocation \(N_r\) in each iteration are carefully chosen to control both the estimation error and the probability of incorrect elimination. These choices are essential for proving that the final recommendation achieves a minimax optimal expected optimality gap.

\begin{algorithm}[t]
\caption{\textsc{QUEST}: \textbf{Qu}antum \textbf{E}fficient \textbf{S}imulation op\textbf{T}imization}
\label{alg:quest}
\begin{algorithmic}[1]
\Require Quantum query budget $T\ge2$, per-estimate failure probability $\delta=1/T$, decision set $\mathcal X$
\State Initialize $\mathcal B^{(1)}\gets \mathcal X$, $\varepsilon_1\gets 1/2$, $r\gets 1$, and $T_0\gets T$
\While{$T_{r-1}>0$ and $|\mathcal B^{(r)}|>1$}
    \State Set $\tilde\epsilon_r\gets \varepsilon_r/2$ \Comment{Target estimation accuracy}
    \State Set $N_r\gets \left\lceil C\tilde\epsilon_r^{-1}\log(1/\delta)\right\rceil$ \Comment{\(C\) depends only on $\sigma$}
    \If{$T_{r-1}<N_r|\mathcal B^{(r)}|$}
        \State \textbf{break}
    \EndIf
    \For{each decision $x\in\mathcal B^{(r)}$}
        \State Run $\textsc{QMC}(\tilde\epsilon_r,\delta)$ with $N_r$ queries to obtain $\hat y^{(r)}(x)$
    \EndFor
    \State Set $\hat y_{\max}^{(r)}\gets \max_{x\in\mathcal B^{(r)}}\hat y^{(r)}(x)$
    \State Update $\mathcal B^{(r+1)}\gets\left\{x\in\mathcal B^{(r)}:\hat y_{\max}^{(r)}-\hat y^{(r)}(x)\le \varepsilon_r\right\}$
    \State Update $T_r\gets T_{r-1}-N_r|\mathcal B^{(r)}|$
    \State Set $\varepsilon_{r+1}\gets \varepsilon_r/2$ and $r\gets r+1$
\EndWhile
\State \Return any decision $\hat x_T\in\mathcal B^{(r)}$
\end{algorithmic}
\end{algorithm}

\subsection{Rate-Optimality Analysis}
We now prove that \textsc{QUEST} attains the minimax rate up to logarithmic factors. For each iteration \(r\) and each active decision \(x\in\mathcal B^{(r)}\), let \(\hat y^{(r)}(x)\) denote the estimate of the mean performance \(y(x)\) returned by the quantum Monte Carlo subroutine in Algorithm~\ref{alg:quest}. The following lemma defines a high-probability good event under which all estimates produced throughout the algorithm are sufficiently accurate. It also shows that the number of completed iterations is at most logarithmic in the query budget \(T\), which enables us to control the total failure probability through a union bound over all iterations and active decisions.

\begin{lemma}
\label{lem:quest_good_event}
Let $R$ be the number of
completed iterations of Algorithm \ref{alg:quest}. Then $R\le \lceil \log_2(T+2)\rceil$. Moreover, the event
\[
    \mathcal G
    :=
    \left\{
    |\hat y^{(r)}(x)-y(x)|
    \le
    \frac{\varepsilon_r}{2},
    \forall r\le R, x\in\mathcal B^{(r)}
    \right\}
\]
satisfies
\[
    \mathbb P(\mathcal G^c)
    \le
    \frac{|\mathcal X|\left\lceil\log_2(T+2)\right\rceil}{T}.
\]
\end{lemma}

The next lemma shows that the elimination rule in \textsc{QUEST} is safe and effective on the good event \(\mathcal G\). Specifically, it guarantees that the true optimal decision \(x^\star\) is never removed from the active set. At the same time, after each completed iteration, all surviving decisions must be close to optimal, with an optimality gap controlled by the current threshold \(\varepsilon_r\). This property is the key step for converting the estimation accuracy of the quantum Monte Carlo subroutine into an upper bound on the final expected optimality gap.

\begin{lemma}
\label{lem:quest_elimination}
On the good event $\mathcal G$, the optimal decision $x^\star$ is never
eliminated. Moreover, after any completed iteration $r$, every surviving decision
$x\in\mathcal B^{(r+1)}$ satisfies
\[
    y(x^\star)-y(x)\le 2\varepsilon_r .
\]
\end{lemma}

Lemma \ref{lem:quest_gap_good} combines the iteration structure of \textsc{QUEST} with the elimination property established above. On the good event \(\mathcal G\), the optimal decision remains active throughout the algorithm, and all surviving decisions become progressively closer to optimal as the threshold decreases. Since each completed iteration consumes a prescribed number of quantum queries, the final threshold can be related to the total budget \(T\). This yields an upper bound on the optimality gap of the final recommendation on the good event, of order \((|\mathcal X|\log T)/T\).

\begin{lemma}
\label{lem:quest_gap_good}
On the good event $\mathcal G$, the decision $\hat x_T$ returned by
Algorithm~\ref{alg:quest} satisfies
\[
    y(x^\star)-y(\hat x_T)
    \le
    C^\prime\frac{|\mathcal X|\log T}{T},
\]
where $C^\prime>0$ is a universal constant.
\end{lemma}

The preceding lemmas provide the key ingredients for bounding the performance of \textsc{QUEST}. On the good event \(\mathcal G\), the elimination rule preserves the optimal decision and ensures that every surviving decision has a small optimality gap. The final gap is therefore controlled by the last accuracy threshold reached before the query budget is exhausted. On the complement event \(\mathcal G^c\), the optimality gap is bounded by a constant because \(y(x)\in[0,1]\). Combining these two cases yields the following upper bound on the expected optimality gap of \textsc{QUEST}.

\begin{theorem}
\label{thm:quest_upper_bound}
For any problem instance $\nu\in\mathcal M$, let $\hat x_T$ be the decision
returned by Algorithm~\ref{alg:quest} after at most $T$ quantum queries. Then
\[
    \mathcal E_T(\textsc{QUEST},\nu)
    =
    O\left(
        \frac{|\mathcal X|\log T}{T}
    \right).
\]
\end{theorem}
\noindent\textit{Proof of Theorem~\ref{thm:quest_upper_bound}.}
We now convert the high-probability guarantee on the good event into an
expected optimality gap bound. By Lemma~\ref{lem:quest_gap_good}, on the good
event $\mathcal G$, the returned decision satisfies
\[
    y(x^\star)-y(\hat x_T)
    \le
    C^\prime\frac{|\mathcal X|\log T}{T}.
\]
On the complement event $\mathcal G^c$, the estimates may fail, and the
algorithm may return a suboptimal decision. However, since $y(x)\in[0,1]$ for
all $x\in\mathcal X$, the optimality gap is always bounded by one:
\[
    y(x^\star)-y(\hat x_T)\le 1.
\]
Therefore, decomposing the expectation according to whether $\mathcal G$ occurs,
we obtain
\[
\begin{aligned}
    \mathcal E_T(\textsc{QUEST},\nu)
    &=
    \mathbb E\!\left[
        y(x^\star)-y(\hat x_T)
    \right] \\
    &=
    \mathbb E\!\left[
        \bigl(y(x^\star)-y(\hat x_T)\bigr)\mathbf 1_{\mathcal G}
    \right]
    +
    \mathbb E\!\left[
        \bigl(y(x^\star)-y(\hat x_T)\bigr)\mathbf 1_{\mathcal G^c}
    \right] \\
    &\le
    C^\prime\frac{|\mathcal X|\log T}{T}
    +
    \mathbb P(\mathcal G^c).
\end{aligned}
\]
The first term corresponds to the controlled gap on the good event, while the
second term accounts for the worst-case loss on the failure event. By
Lemma~\ref{lem:quest_good_event},
\[
    \mathbb P(\mathcal G^c)
    \le
    \frac{|\mathcal X|\left\lceil\log_2(T+2)\right\rceil}{T}.
\]
Combining the two bounds and absorbing constants, including the change of
logarithm base, into a universal constant gives
\[
    \mathcal E_T(\textsc{QUEST},\nu)
    =
    O\left(\frac{|\mathcal X|\log T}{T}\right).
\]
\qed
\begin{remark}
By Theorem~\ref{thm:quest_upper_bound}, \textsc{QUEST} achieves an expected optimality gap of order \(O(|\mathcal X|\log T/T)\). Comparing this upper bound with the quantum minimax lower bound \(\Omega(|\mathcal X|/T)\) in Theorem~\ref{thm: quantum lower bound}, we see that \textsc{QUEST} matches the minimax lower bound up to a logarithmic factor. Since \(\log T\) grows slowly relative to the polynomial dependence on \(T\), we regard \textsc{QUEST} as minimax rate optimal up to logarithmic factors, or equivalently,
\[
    \sup_{\nu\in\mathcal M}
    \mathcal E_T(\textsc{QUEST},\nu)
    =
    \widetilde O\left(\frac{|\mathcal X|}{T}\right).
\]
\end{remark}

\section{Numerical Experiments}
\label{sec: numerical}
In this section, we present numerical experiments to evaluate the practical performance of \textsc{QUEST} for discrete simulation optimization. We examine the expected optimality gap under varying simulation budgets, problem sizes, and performance distributions.

To enable a fair comparison and assess the potential quantum advantage, we introduce a classical counterpart of \textsc{QUEST}, denoted by \textsc{MC-EST}. In this algorithm, the quantum estimation component is replaced by standard Monte Carlo sampling, while the same slow-measurement framework is preserved.

We implement \textsc{QUEST} using Qiskit-based simulation. Throughout this section, all performance metrics are estimated from 100 independent replications, and the corresponding 95\% confidence intervals are reported.

In the first example, we consider a finite decision set with \(|\mathcal X|=20\). For each decision \(x\in\mathcal X\), the performance observation follows a Bernoulli distribution. The true mean performances range from 0.66 to 0.85 and are closely spaced, making the decision alternatives difficult to distinguish. This setting provides a nontrivial discrete simulation optimization problem with small optimality gaps among candidate decisions.

\begin{figure}
\centering
\includegraphics[width=0.8\linewidth]{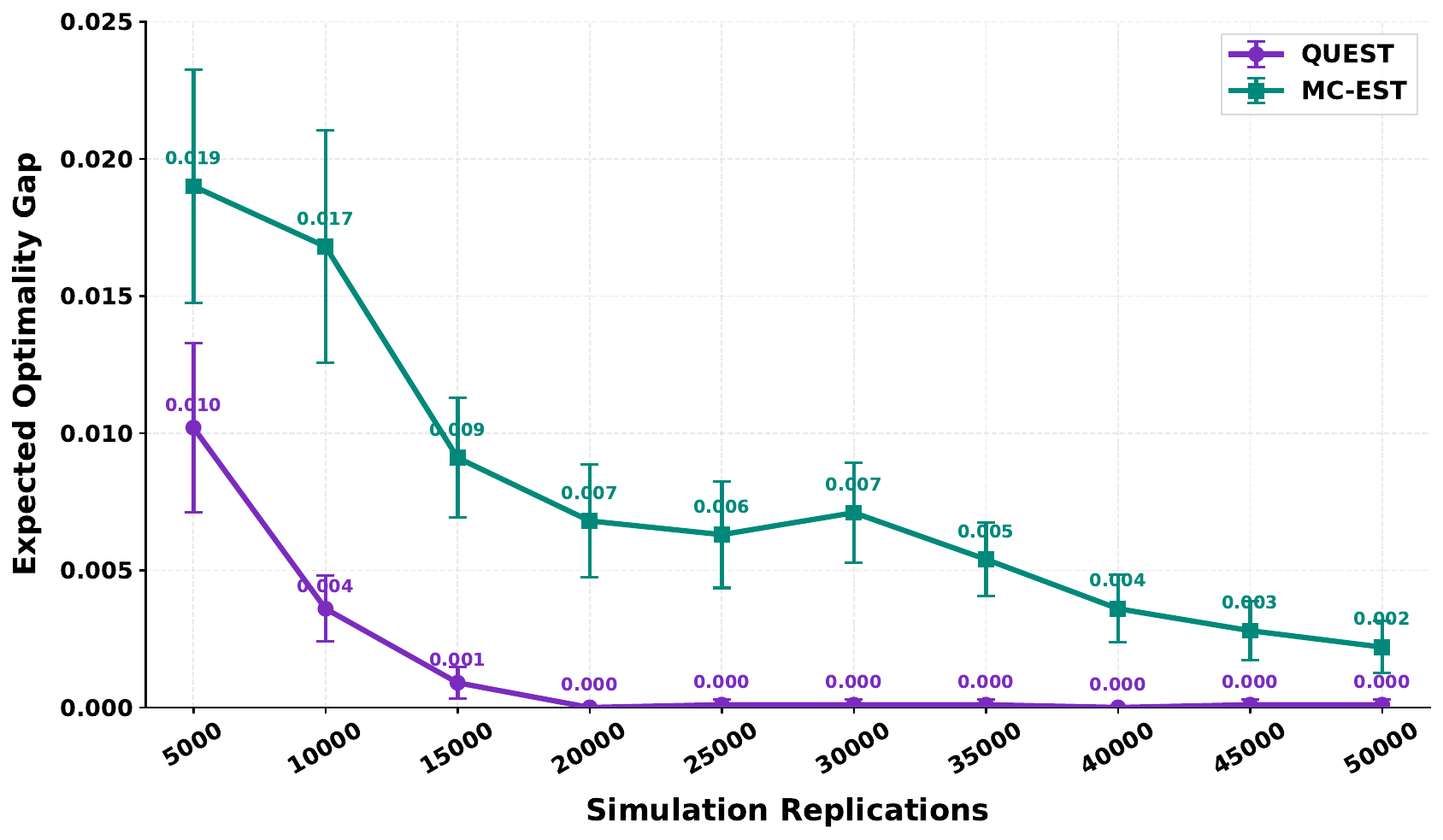}
\caption{Expected Optimality Gap under Bernoulli Distribution}
\label{fig:bernoulli_gap}
\end{figure}

Figure~\ref{fig:bernoulli_gap} reports the expected optimality gap under the Bernoulli performance setting as the simulation budget \(T\) increases. Both \textsc{QUEST} and \textsc{MC-EST} reduce the expected optimality gap as more simulation replications are used, but \textsc{QUEST} exhibits a much faster decrease. The expected optimality gap of \textsc{QUEST} drops sharply from about \(0.010\) at \(T=5000\) to a negligible level around \(T=20000\). In contrast, \textsc{MC-EST} decreases more gradually, from about \(0.019\) at \(T=5000\) to about \(0.002\) at \(T=50000\). These results indicate that, under the same simulation budget, \textsc{QUEST} identifies higher-quality decisions more quickly than its classical Monte Carlo counterpart, highlighting its advantage in convergence speed.

We next evaluate the performance of \textsc{QUEST} and \textsc{MC-EST} under varying problem sizes. In this experiment, the number of decision alternatives is set to \(|\mathcal X|\in\{10,20,30\}\), while the total simulation budget is fixed at \(T=10000\) for each independent run.
We construct the true means so that instances with larger \(|\mathcal X|\) contain more densely spaced near-optimal competitors. Performance observations are generated from Bernoulli distributions with these prescribed means. This setup is intentionally nontrivial: as \(|\mathcal X|\) increases, the identification task becomes harder because more alternatives have similar true performance, making incorrect final selection more likely under a finite budget. For each setting, we report the final expected optimality gap to measure the quality of the final recommendation.

\begin{figure}
\centering
\includegraphics[width=1.0\linewidth]{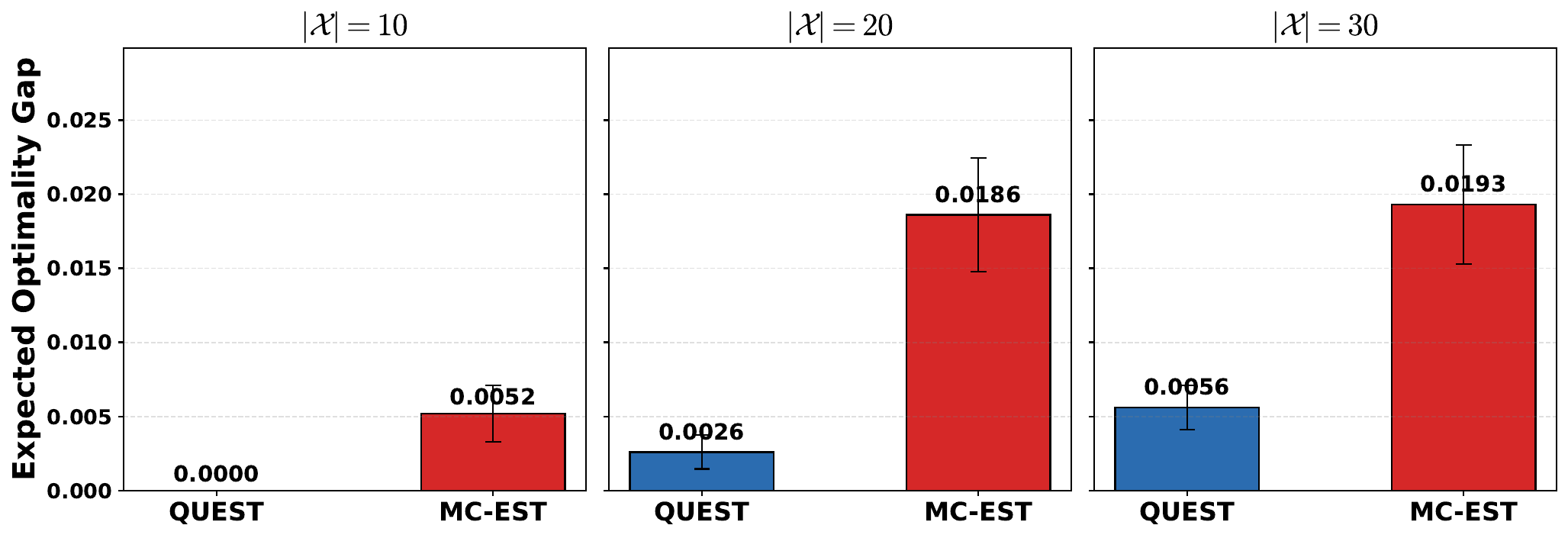}
\caption{Expected Optimality Gap under Different Problem Sizes ($T=10000$)}
\label{fig:problem sizes}
\end{figure}

Figure~\ref{fig:problem sizes} compares the final expected optimality gaps of \textsc{QUEST} and \textsc{MC-EST} under different problem sizes. Across all settings, \textsc{QUEST} achieves a substantially smaller expected optimality gap than \textsc{MC-EST}. When \(|\mathcal X|=10\), the gap of \textsc{QUEST} is essentially zero, while \textsc{MC-EST} has a positive gap of about \(0.0052\). As the problem size increases, the identification task becomes more challenging, but \textsc{QUEST} continues to maintain small gaps of about \(0.0026\) and \(0.0056\) for \(|\mathcal X|=20\) and \(|\mathcal X|=30\), respectively. In contrast, \textsc{MC-EST} exhibits much larger gaps of about \(0.0186\) and \(0.0193\). These results show that, under the same simulation budget, \textsc{QUEST} consistently identifies higher-quality decisions across different problem sizes, which provides empirical evidence of the quantum advantage in discrete simulation optimization.

We further compare the performance of \textsc{QUEST} and \textsc{MC-EST} under different performance distributions. In this experiment, we fix the number of decisions at \(|\mathcal X|=20\) and evaluate both methods under two fixed simulation budgets, \(T=10000\) and \(T=20000\).
We consider three observation models: Gaussian, Student-\(t\), and Beta. The Gaussian and Student-\(t\) models represent continuous random performance, with the Student-\(t\) distribution introducing heavier tails. The Beta model provides a bounded \([0,1]\) setting with controllable variability. For \textsc{QUEST}, continuous outcomes are mapped to a finite quantum support through a uniform grid approximation with \(q=4\) grid qubits, corresponding to \(2^q=16\) support points, when constructing the quantum oracle. This experiment allows us to compare decision quality and assess the robustness of both methods across different performance distributions.

\begin{figure}
\centering
\includegraphics[width=1.0\linewidth]{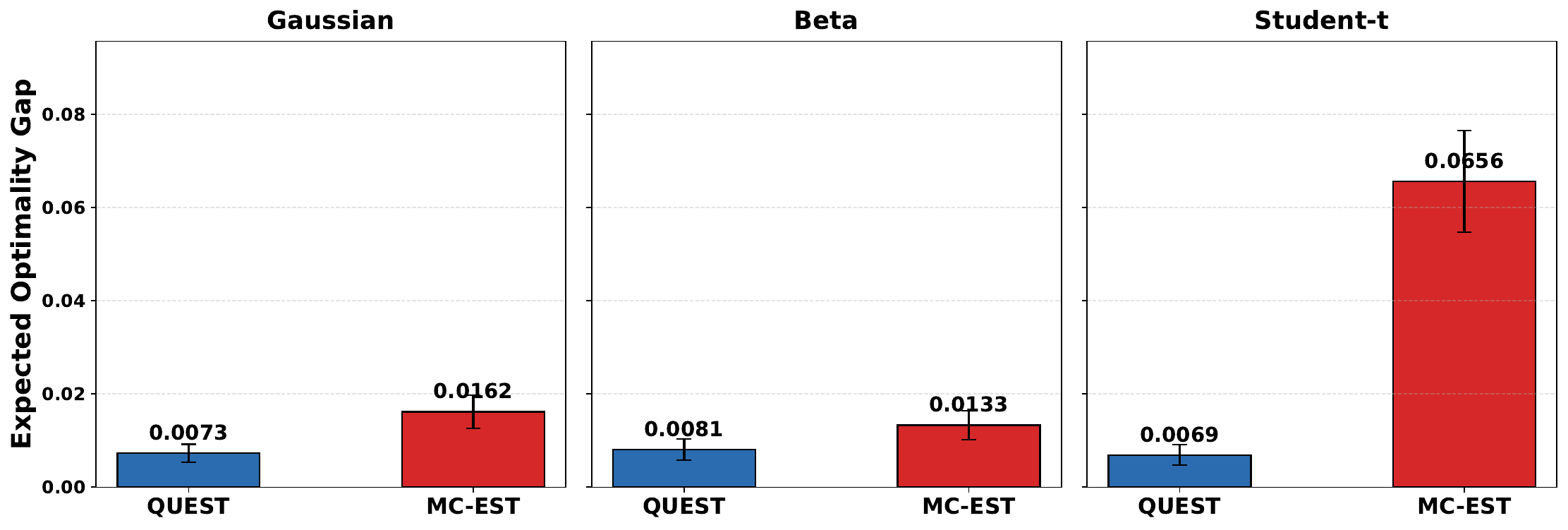}
\caption{Expected Optimality Gap under Different Performance Distributions ($T=10000$)}
\label{fig:distribution}
\end{figure}

\begin{figure}
\centering
\includegraphics[width=1.0\linewidth]{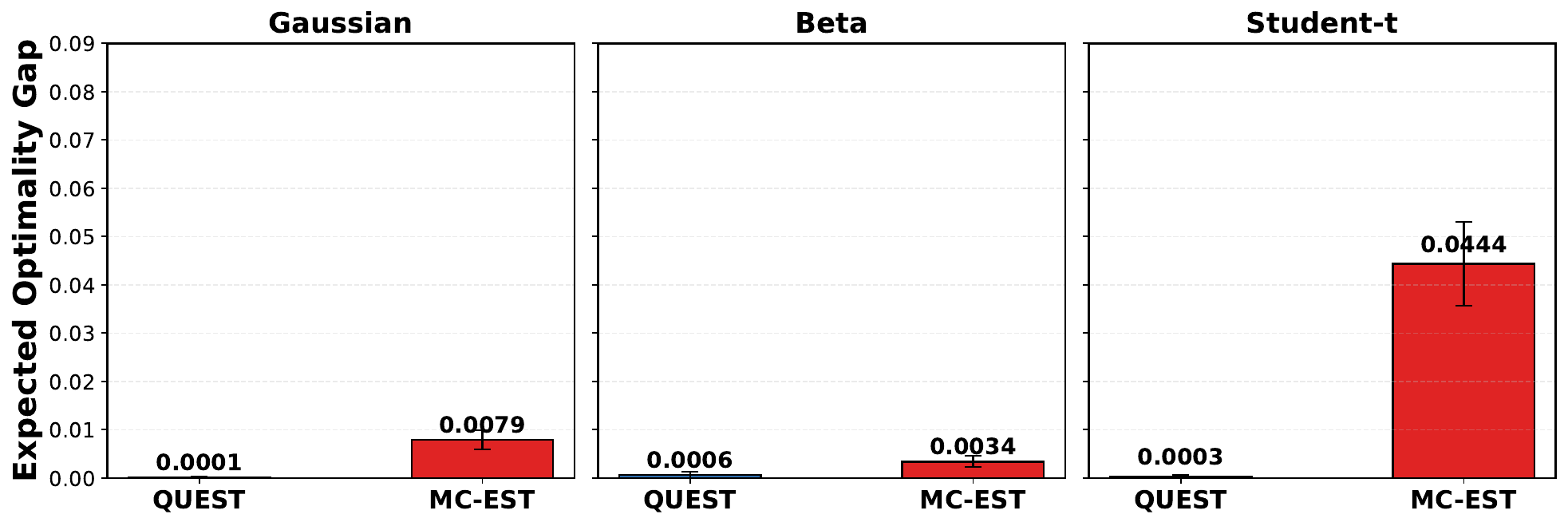}
\caption{Expected Optimality Gap under Different Performance Distributions ($T=20000$)}
\label{fig:distribution2}
\end{figure}

Figures~\ref{fig:distribution} and~\ref{fig:distribution2} compare the final expected optimality gaps of \textsc{QUEST} and \textsc{MC-EST} under Gaussian, Beta, and Student-\(t\) performance distributions. When \(T=10000\), \textsc{QUEST} achieves smaller expected optimality gaps across all three distributions, with gaps around \(0.007\)--\(0.008\). In contrast, \textsc{MC-EST} has larger gaps, especially under the Student-\(t\) distribution, where the gap reaches about \(0.0656\). When the budget increases to \(T=20000\), the advantage of \textsc{QUEST} becomes more pronounced. Its expected optimality gaps are nearly zero under all three distributions, while \textsc{MC-EST} still exhibits positive gaps, particularly under the heavy-tailed Student-\(t\) setting. These results show that \textsc{QUEST} consistently identifies higher-quality decisions than its classical Monte Carlo counterpart under the same simulation budget. The strong performance across Gaussian, Beta, and Student-\(t\) observations provides empirical evidence of the quantum advantage and suggests that \textsc{QUEST} remains robust across different performance distributions, including settings beyond the light-tailed assumptions used in the theoretical analysis.

\section{Conclusion}
\label{sec: conclusion}
This paper studies how quantum computing can accelerate the convergence rate of discrete simulation optimization. We establish a minimax lower bound on the expected optimality gap for quantum algorithms and show the resulting quantum advantage by comparison with the classical lower bound. We further propose \textsc{QUEST}, a quantum discrete simulation optimization algorithm, and prove that it is minimax rate optimal up to logarithmic factors. Numerical experiments support the effectiveness of the proposed algorithm. Future directions include extending the framework to continuous simulation optimization and more complex settings, such as nested simulation optimization and risk-measure-based objectives.

Overall, this paper develops a theoretical framework for analyzing quantum algorithms for discrete simulation optimization. Its problem formulation, performance metrics, and analytical tools can potentially be extended to other fixed-budget simulation optimization settings. More broadly, this work opens a new direction for using quantum computing to accelerate simulation optimization.


%
%
%



\bibliographystyle{informs2014} 
\bibliography{sample} 

@book{lattimore2020bandit,
  title={Bandit algorithms},
  author={Lattimore, Tor and Szepesv{\'a}ri, Csaba},
  year={2020},
  publisher={Cambridge University Press}
}

@article{chen2000simulation,
  title={Simulation budget allocation for further enhancing the efficiency of ordinal optimization},
  author={Chen, Chun-Hung and Lin, Jianwu and Y{\"u}cesan, Enver and Chick, Stephen E},
  journal={Discrete Event Dynamic Systems},
  volume={10},
  pages={251--270},
  year={2000},
  publisher={Springer}
}

@article{shen2021ranking,
  title={Ranking and selection with covariates for personalized decision making},
  author={Shen, Haihui and Hong, L Jeff and Zhang, Xiaowei},
  journal={INFORMS Journal on Computing},
  volume={33},
  number={4},
  pages={1500--1519},
  year={2021},
  publisher={INFORMS}
}

@article{kim2001fully,
  title={A fully sequential procedure for indifference-zone selection in simulation},
  author={Kim, Seong-Hee and Nelson, Barry L},
  journal={ACM Transactions on Modeling and Computer Simulation (TOMACS)},
  volume={11},
  number={3},
  pages={251--273},
  year={2001},
  publisher={ACM New York, NY, USA}
}

@book{nielsen2010quantum,
  title={Quantum computation and quantum information},
  author={Nielsen, Michael A and Chuang, Isaac L},
  year={2010},
  publisher={Cambridge university press}
}

@inproceedings{wang2025quantum,
  title={Quantum Best Arm Identification with Quantum Oracles},
  author={Wang, Xuchuang and Chen, Yu-Zhen Janice and de Andrade, Matheus Guedes and Allcock, Jonathan and Hajiesmaili, Mohammad and Lui, John CS and Towsley, Don},
  booktitle={Proceedings of the AAAI Conference on Artificial Intelligence},
  volume={39},
  number={20},
  pages={21321--21329},
  year={2025}
}

@article{montanaro2015quantum,
  title={Quantum speedup of Monte Carlo methods},
  author={Montanaro, Ashley},
  journal={Proceedings of the Royal Society A: Mathematical, Physical and Engineering Sciences},
  volume={471},
  number={2181},
  pages={20150301},
  year={2015},
  publisher={The Royal Society}
}

@article{esa2024single,
  title={SINGLE SERVER BATCH SERVICE MARKOVIAN (M/M (a, b)/1) QUEUE WITH VARIANT WORKING VACATION},
  author={ESA, SHERIF ABRAHIM},
  year={2024}
}

@article{dunjko2016quantum,
  title={Quantum-enhanced machine learning},
  author={Dunjko, Vedran and Taylor, Jacob M and Briegel, Hans J},
  journal={Physical review letters},
  volume={117},
  number={13},
  pages={130501},
  year={2016},
  publisher={APS}
}

@article{lee2012approximate,
  title={Approximate simulation budget allocation for selecting the best design in the presence of stochastic constraints},
  author={Lee, Loo Hay and Pujowidianto, Nugroho Artadi and Li, Ling-Wei and Chen, Chun-Hung and Yap, Chee Meng},
  journal={IEEE Transactions on Automatic Control},
  volume={57},
  number={11},
  pages={2940--2945},
  year={2012},
  publisher={IEEE}
}

@article{hu2024multi,
  title={Multi-Task Best Arm Identification with Risk Constraints},
  author={Hu, Mingjie and Hu, Jianqiang},
  journal={Available at SSRN 5214504},
  year={2024}
}

@article{casale2020quantum,
  title={Quantum bandits},
  author={Casal{\'e}, Balthazar and Di Molfetta, Giuseppe and Kadri, Hachem and Ralaivola, Liva},
  journal={Quantum Machine Intelligence},
  volume={2},
  number={1},
  pages={11},
  year={2020},
  publisher={Springer}
}

@inproceedings{wang2021quantum,
  title={Quantum exploration algorithms for multi-armed bandits},
  author={Wang, Daochen and You, Xuchen and Li, Tongyang and Childs, Andrew M},
  booktitle={Proceedings of the AAAI Conference on Artificial Intelligence},
  volume={35},
  number={11},
  pages={10102--10110},
  year={2021}
}

@article{meyer2022survey,
  title={A survey on quantum reinforcement learning},
  author={Meyer, Nico and Ufrecht, Christian and Periyasamy, Maniraman and Scherer, Daniel D and Plinge, Axel and Mutschler, Christopher},
  journal={arXiv preprint arXiv:2211.03464},
  year={2022}
}

@article{hu2025optimal,
  title={Optimal computation budget allocation with Gaussian process regression},
  author={Hu, Mingjie and Xu, Jie and Chen, Chun-Hung and Hu, Jian-Qiang},
  journal={European Journal of Operational Research},
  volume={322},
  number={1},
  pages={147--156},
  year={2025},
  publisher={Elsevier}
}

@article{wu2024data,
  title={Data-driven ranking and selection under input uncertainty},
  author={Wu, Di and Wang, Yuhao and Zhou, Enlu},
  journal={Operations Research},
  volume={72},
  number={2},
  pages={781--795},
  year={2024},
  publisher={INFORMS}
}

@article{abbas2024challenges,
  title={Challenges and opportunities in quantum optimization},
  author={Abbas, Amira and Ambainis, Andris and Augustino, Brandon and B{\"a}rtschi, Andreas and Buhrman, Harry and Coffrin, Carleton and Cortiana, Giorgio and Dunjko, Vedran and Egger, Daniel J and Elmegreen, Bruce G and others},
  journal={Nature Reviews Physics},
  volume={6},
  number={12},
  pages={718--735},
  year={2024},
  publisher={Nature Publishing Group}
}

@article{hu2025quantum,
  title={Quantum Ranking and Selection under Stochastic Constraints},
  author={Hu, Mingjie and Hu, Jianqiang and Wang, Tianxiang},
  note={Available at SSRN 5948136},
  year={2025}
}

@article{javadi2024quantum,
  title={Quantum computing with Qiskit},
  author={Javadi-Abhari, Ali and Treinish, Matthew and Krsulich, Kevin and Wood, Christopher J and Lishman, Jake and Gacon, Julien and Martiel, Simon and Nation, Paul D and Bishop, Lev S and Cross, Andrew W and others},
  journal={arXiv preprint arXiv:2405.08810},
  year={2024}
}

@inproceedings{kothari2023mean,
  title={Mean estimation when you have the source code; or, quantum Monte Carlo methods},
  author={Kothari, Robin and O'Donnell, Ryan},
  booktitle={Proceedings of the 2023 Annual ACM-SIAM Symposium on Discrete Algorithms (SODA)},
  pages={1186--1215},
  year={2023},
  organization={SIAM}
}

@article{blanchet2025quantum,
  title={Quantum speedup of non-linear Monte Carlo problems},
  author={Blanchet, Jose and Hamoudi, Yassine and Szegedy, Mario and Wang, Guanyang},
  journal={arXiv preprint arXiv:2502.05094},
  year={2025}
}

@article{dai2023quantum,
  title={Quantum bayesian optimization},
  author={Dai, Zhongxiang and Lau, Gregory Kang Ruey and Verma, Arun and Shu, Yao and Low, Bryan Kian Hsiang and Jaillet, Patrick},
  journal={Advances in neural information processing systems},
  volume={36},
  pages={20179--20207},
  year={2023}
}

@inproceedings{wang2021quantum2,
  title={Quantum algorithms for reinforcement learning with a generative model},
  author={Wang, Daochen and Sundaram, Aarthi and Kothari, Robin and Kapoor, Ashish and Roetteler, Martin},
  booktitle={International Conference on Machine Learning},
  pages={10916--10926},
  year={2021},
  organization={PMLR}
}

@inproceedings{zabinsky2025comprehensive,
  title={A Comprehensive Approach to Optimizing Complex Systems},
  author={Zabinsky, Zelda B and Maneekul, Pariyakorn},
  booktitle={2025 Winter Simulation Conference (WSC)},
  pages={43--57},
  year={2025},
  organization={IEEE}
}

@inproceedings{gacon2020quantum,
  title={Quantum-enhanced simulation-based optimization},
  author={Gacon, Julien and Zoufal, Christa and Woerner, Stefan},
  booktitle={2020 IEEE International conference on quantum computing and engineering (QCE)},
  pages={47--55},
  year={2020},
  organization={IEEE}
}

@article{bulger2003implementing,
  title={Implementing pure adaptive search with Grover's quantum algorithm},
  author={Bulger, David and Baritompa, William P and Wood, Graham R},
  journal={Journal of optimization theory and applications},
  volume={116},
  number={3},
  pages={517--529},
  year={2003},
  publisher={Springer}
}

@article{baritompa2005grover,
  title={Grover's quantum algorithm applied to global optimization},
  author={Baritompa, William P and Bulger, David W and Wood, Graham R},
  journal={SIAM Journal on Optimization},
  volume={15},
  number={4},
  pages={1170--1184},
  year={2005},
  publisher={SIAM}
}

@article{gilliam2021grover,
  title={Grover adaptive search for constrained polynomial binary optimization},
  author={Gilliam, Austin and Woerner, Stefan and Gonciulea, Constantin},
  journal={Quantum},
  volume={5},
  pages={428},
  year={2021},
  publisher={Verein zur F{\"o}rderung des Open Access Publizierens in den Quantenwissenschaften}
}

@article{carrera2021efficient,
  title={Efficient state preparation for quantum amplitude estimation},
  author={Carrera Vazquez, Almudena and Woerner, Stefan},
  journal={Physical Review Applied},
  volume={15},
  number={3},
  pages={034027},
  year={2021},
  publisher={APS}
}

@article{gao2017new,
  title={A new budget allocation framework for the expected opportunity cost},
  author={Gao, Siyang and Chen, Weiwei and Shi, Leyuan},
  journal={Operations Research},
  volume={65},
  number={3},
  pages={787--803},
  year={2017},
  publisher={Informs}
}

@article{wang2025optimal,
  title={Optimal computing budget allocation for data-driven ranking and selection},
  author={Wang, Yuhao and Zhou, Enlu},
  journal={INFORMS Journal on Optimization},
  volume={7},
  number={1},
  pages={1--19},
  year={2025},
  publisher={INFORMS}
}

@article{hu2026quantum,
  title={Quantum Grover Adaptive Search for Discrete Simulation Optimization},
  author={Hu, Mingjie and Hu, Jian-qiang and Zhou, Enlu},
  journal={arXiv preprint arXiv:2604.26277},
  year={2026}
}

@inproceedings{wan2023quantum,
  title={Quantum multi-armed bandits and stochastic linear bandits enjoy logarithmic regrets},
  author={Wan, Zongqi and Zhang, Zhijie and Li, Tongyang and Zhang, Jialin and Sun, Xiaoming},
  booktitle={Proceedings of the AAAI Conference on Artificial Intelligence},
  volume={37},
  number={8},
  pages={10087--10094},
  year={2023}
}

@inproceedings{grover1996fast,
  title={A fast quantum mechanical algorithm for database search},
  author={Grover, Lov K},
  booktitle={Proceedings of the twenty-eighth annual ACM symposium on Theory of computing},
  pages={212--219},
  year={1996}
}

@article{durr1996quantum,
  title={A quantum algorithm for finding the minimum},
  author={Durr, Christoph and Hoyer, Peter},
  journal={arXiv preprint quant-ph/9607014},
  year={1996}
}

@article{mande2026tight,
  title={Tight bounds for quantum phase estimation and related problems},
  author={Mande, Nikhil S and de Wolf, Ronald},
  journal={Quantum},
  volume={10},
  pages={2140},
  year={2026},
  publisher={Verein zur F{\"o}rderung des Open Access Publizierens in den Quantenwissenschaften}
}





  



\end{document}